\documentclass{SCIS2018}
\usepackage{makecell}
\usepackage{tikz}

\begin{document}

\ArticleType{RESEARCH PAPER}
\Year{2018}
\Month{}
\Vol{61}
\No{}
\DOI{}
\ArtNo{}
\ReceiveDate{}
\ReviseDate{}
\AcceptDate{}
\OnlineDate{}
\title {On Sub-Packetization and Access Number of Capacity-Achieving PIR Schemes for MDS Coded Non-Colluding Servers}{}
\author[1,2]{Jingke XU}{}
\author[1,2]{Zhifang ZHANG}{{zfz@amss.ac.cn}}

\AuthorMark{ Xu J K}

\AuthorCitation{ Xu J K,  Zhang Z F}

\address[1]{Key Laboratory of Mathematics Mechanization,
Academy of Mathematics and Systems Science, \\ Chinese Academy of Sciences, Beijing {\rm 100190}, China }
\address[2]{School of Mathematical Sciences, University of Chinese Academy of Sciences, Beijing {\rm 100049}, China}

\abstract{Consider the problem of private information retrieval (PIR) over a distributed storage system where $M$ records are stored across $N$ servers by using an $[N,K]$ MDS code. For simplicity, this problem is usually referred  as the coded PIR problem. In 2016, Banawan and Ulukus designed the first capacity-achieving coded PIR scheme with sub-packetization $KN^{M}$ and access number $MKN^{M}$, where capacity characterizes the minimal download size for retrieving per unit of data, and sub-packetization and access number are two metrics closely related to implementation complexity.
In this paper, we focus on minimizing the sub-packetization and the access number for linear capacity-achieving  coded PIR schemes. We first determine the lower bounds on sub-packetization and  access number, which are $Kn^{M-1}$ and $MKn^{M-1}$, respectively, in the nontrivial cases (i.e. $N\!>\!K\!\geq\!1$ and $M\!>\!1$), where $n\!=\!N/{\rm gcd}(N,K)$. We then design a general linear capacity-achieving coded PIR scheme to simultaneously attain these two bounds, implying tightness of both bounds.}


\maketitle{}

\section{Introduction}\label{sec1}
Private information retrieval (PIR) is a canonical problem in the study of privacy issues that arise from the  retrieval of information from public databases. Specifically, PIR involves a database that contains $M$ records and a user’s with query interest $\theta\in\{1,...,M\}$. The goal is to make  the user get retrieves the $\theta$th record without revealing the index $\theta$. In the information theoretic sense, the PIR problem can only be solved trivially solved by downloading all $M$ records if the database is stored in one server. Therefore, in FOCS'95
Chor et al. \cite{CKGS95FOCS:PIR,CKGS98FOCS:PIR} developed the distributed formulation for of PIR, where the database is stored across $N$ servers and the user can communicate with all $N$ servers. The privacy requirement is to ensure the secrecy of $\theta$  against any individual server. Since then, PIR has become a central research topic in the computer science literature,
see \cite{Gas04EATCS:SurveyonPIR} for a survey on PIR.

A central issue in PIR is minimizing the communication cost, which is usually measured by the total number of bits transferred from the user to the servers (i.e. the upload size) and  from the servers to the user (i.e. the download size). In the initial setting of PIR where each record is set to one bit, the minimum communication cost achieved is $M^{O(\frac{1}{\log\log M})}$ \cite{Dvir&Gopi15STOC:2PIR,Efremenko09:LDC}. However, in real-world applications, it is common for the size of each record  to be arbitrarily large. Therefore, the upload size is usually negligible compared to the download size. Consequently, the communication cost can be measured by considering only the download size. Specifically, define the {\it rate} of a PIR scheme as the ratio between the size of the retrieved record and the download size, and define the {\it capacity} as the supremum of the rate over all PIR schemes. In addition, the reciprocal of the capacity describes the minimum possible download size per unit of retrieved records. Recently, much work has been done on determining the capacity of PIR in various cases:

{\it Replication-based PIR: } In this case, each of the $N$ servers stores a replication of the database. In \cite{Sun&Jafar16:CapacityPIR}, Sun and Jafar proved that the capacity in the non-colluding case is {\small$(1+\frac{1}{N}+\cdots+\frac{1}{N^{M-1}})^{-1}$}. In \cite{Sun&Jafar16:ColludPIR}, they derived the capacity for the colluding case (i.e. ensuring the secrecy of the retrieval index $\theta$ against any subset containing at most $T$ colluding servers for $1\!\leq\!T\!<\!N$) and the robust case (i.e. some servers may fail to respond). They also determined the capacity of PIR with symmetric privacy in \cite{Sun&Jafar16:CapaSymmPIR}, where symmetric privacy means that the user is  required to get no information about the record other than the $\theta$th record. Banawan and Ulukus ~\cite{Bana&Uluk17:CapacityMPIR} recently derived the capacity of a multi-message PIR with replicated non-colluding servers for the case of retrieving more than half records. In \cite{Bana&Uluk17:CapacityBTPIR}, they studied the capacity of PIR with colluding and Byzantine servers. Other studies considered the PIR problem when some side information is available to the user ~\cite{Tandom17:CapaCAPIR,YWKBUlukusLofCAPIT:CapaCAPIR,KGHERS:CapaCAPIR,Chen&Wang&Jafar17:CapacityPIRPSI}.

{\it Coded PIR:} In this case, the database is stored across $N$ servers using some code. In particular, an $[N,K]$ MDS code is mostly used. Banawan and Ulukus \cite{Bana&Uluk16:CapacityPIRCoded} proved that the capacity of the PIR problem with  MDS coded non-colluding servers (i.e. coded PIR) is {\small$(1+\frac{K}{N}+\cdots+\frac{K}{N^{M-1}})^{-1}$}. In \cite{TR16:MDSPIR}, the authors designed a scheme for MDS coded non-colluding servers with rate $1-\frac{K}{N}$.
 The capacity of PIR with symmetric privacy  based on MDS coded non-colluding servers was derived in \cite{Wang&Sko16:CapaSymmePIR}. In \cite{HGHK16:CodedTPIR}, the authors presented a framework for PIR from Reed-Solomon coded colluding servers, and designed a scheme with the rate $1-\frac{K+T-1}{N}$. Another PIR scheme for MDS coded colluding servers was later presented in \cite{YZhangGGe17:CodedTPIR} with the rate {\small$(1+r+\cdots+r^{M-1})^{-1}$}, where $r=1-\binom{N-T}{K}/\binom{N}{K}$.  It remains an open problem to determine the capacity of PIR based on MDS coded colluding servers.

Determining PIR capacity is usually accomplished in two ways: proving an upper bound on the capacity and designing a general PIR scheme with rate attaining the upper bound. Therefore, these schemes are called capacity-achieving PIR schemes. Almost all existing capacity-achieving PIR schemes are implemented by dividing each record into sub-packets (say, $L$ sub-packets) and querying some linear combinations of the sub-packets from each server. We call  $L$ as the sub-packetization of the scheme and call  the total number of sub-packets accessed by all $N$ servers as the access number. Although large sub-packetization helps to improve the PIR rate, it also increases complexity in implementation because larger sub-packetization means more combinations, and thus therefore more multiplications are involved.
The problems of reducing sub-packetization and the access-optimal property have been studied in depth in the literature of minimum storage regenerating codes \cite{Ye&Barg17:MSRsub,B&Kumar17:MSRsub}. However, for the PIR problem, most known capacity-achieving PIR schemes with asymmetric privacy have demonstrate exponential sub-packetization and access number. For example, the capacity-achieving scheme in \cite{Sun&Jafar16:CapacityPIR} exhibits sub-packetization $N^M$ and access number $MN^M$,  and the scheme in \cite{Bana&Uluk16:CapacityPIRCoded} has  sub-packetization $KN^M$ and access number $MKN^M$. On the other hand, a scheme with sub-packetization $K(N-K)$  was designed  in \cite{TR16:MDSPIR} at the sacrifice of failing to achieve the desired of capacity. Theoretically, it is meaningful to  characterize the minimum sub-packetization for achieving capacity in linear PIR schemes.

Our research interest is to minimize both the sub-packetization and the access number for linear capacity-achieving PIR schemes. For replication-based PIR, the paper \cite{Sun&Jafar16:OptimalPIR} first characterized the optimal download cost for arbitrary record length and demonstrated that the optimal sub-packetization for $T=1$ is $N^{M-1}$. One of our recent work \cite{Zhang&Xu17:OptimalSubpacketization} extends this result to general $T$ and proves that the optimal sub-packetization for capacity-achieving PIR schemes over replicated servers is $dn^{M-1}$, where $d={\rm gcd}(N,T),n=N/d$. In this paper, we focus on the sub-packetization and the access number for linear capacity-achieving PIR schemes over MDS coded non-colluding servers. Our contributions are three-fold:
  \begin{enumerate}
  \item A lower bound on the sub-packetization $L$, i.e., $L\geq Kn^{M-1}$, where $n=N/{{\rm gcd}(N,K)}$.
  \item A lower bound on the access number $\omega$, i.e., $\omega \geq MKn^{M-1}$.
  \item A general linear capacity-achieving coded PIR scheme with sub-packetization $L=Kn^{M-1}$ and access number $\omega=MKn^{M-1}$, which implies that our lower bounds are both tight. In other words, we design a capacity-achieving PIR scheme that simultaneously achieves the optimal sub-packetization and the optimal access number.
  \end{enumerate}

The rest of this paper is organized as follows.
First, a formal description of the coded PIR model and a brief recall of the proof for capacity are provided in Section \ref{sec2}. Lower bounds on the sub-packetization and the access number are then presented in Section \ref{sec3} and Section \ref{sec4}, respectively. Finally, a general linear capacity-achieving coded PIR scheme that simultaneously attains the two lower bounds is presented in Section \ref{sec5}.

\section{Preliminaries}\label{sec2}
\subsection{Notations and the PIR model}\label{sec2a}
For positive integers  $m,n\in\mathbb{N}$ with $m<n$, we denote by $[m:n]$  the set $\{m,m+1,...,n\}$ and denote  by $[n]$ the set $\{1,2,...,n\}$.  For a vector $Q=(q_1,...,q_n)$ and any subset $\Gamma=\{i_1,...,i_m\}\subseteq [n]$, let $Q_\Gamma=(q_{i_1},...,q_{i_m})$.
Moreover, to differentiate indices of the servers for records, we use superscripts as indices of the servers and subscripts for the records. For example, we use $Q_\theta^{(i)}$ to denote the a query to the $i$th server when the user wants  the $\theta$th record.
Throughout the paper, we use cursive capital letters to denote random variables such as $\mathcal{W}, \mathcal{Q}$, etc.

Suppose there are $M$ records denoted by $\mathcal{W}_1,\dots,\mathcal{W}_M$. Each record consists of $L$~symbols drawn independently and uniformly from the finite field $\mathbb{F}_q$, i.e.,
\begin{equation}\label{q1}
\forall i\in[M],~H(\mathcal{W}_i)=L,~~~H(\mathcal{W}_1,...,\mathcal{W}_M)=\sum_{i=1}^MH(\mathcal{W}_i)=ML.\;
\end{equation}
where $H(\cdot)$ denotes the entropy function with base $q$.

Moreover, the $M$ records are stored across $N$ servers through an $[N,K]$ MDS code. Therefore, we further assume $L=K\tilde{L}$ and $\mathcal{W}_j\in \mathbb{F}^{K\times\tilde{L}}_q$ for all $j\in[M]$. Let $G=({\bf g}_1,{\bf g}_2,...,{\bf g}_N)\in\mathbb{F}^{K\times N}_q$ be a generator matrix of an $[N,K]$ MDS code over $\mathbb{F}_q$.
Then following MDS encoding, the $i$th server, $\rm{Serv}^{(i)}$, $1\leq i\leq N$, stores
$\mathcal{C}^{(i)}=(\mathcal{C}^{(i)}_1,...,\mathcal{C}^{(i)}_M)$ where
$$\mathcal{C}^{(i)}_j={\bf g}^\tau_i\mathcal{W}_j\in\mathbb{F}_q^{1\times\tilde{L}}{\rm~~for~~}1\leq j\leq M\;.$$
Because of the $[N,K]$ MDS encoding, for any subset $\Gamma\subseteq [N]$ with $|\Gamma|=K$, we have
$H(\mathcal{C}^\Gamma)=ML$ and $H(\mathcal{W}_{[M]}|\mathcal{C}^\Gamma)=0$,
where $\mathcal{C}^\Gamma=\{\mathcal{C}^{(i)}\mid i\in\Gamma\}$.

A PIR scheme allows a user to retrieve a record, say $\mathcal{W}_\theta$, for some $\theta\in[M]$ by accessing the $N$ servers while ensuring the secrecy of the index $\theta$ against any individual server.
 PIR consists of two phases:
 \begin{itemize}
  \item {\bf Query phase.}
    Given an index $\theta\in[M]$ and some random resources $\mathcal{S}$, the user computes ${\rm Que}(\theta,\mathcal{S})=(\mathcal{Q}_\theta^{(1)},...,\mathcal{Q}_\theta^{(N)})$, and sends $\mathcal{Q}_\theta^{(i)}$ to ${\rm Serv}^{(i)}$ for $1\leq i\leq N$. Note that $\mathcal{S}$ and $\theta$ are private information only known  to the user,
  and the function ${\rm Que}(\cdot,\cdot)$ is the {\it query function} determined by the scheme.
   For simplicity, we define the query set $\mathcal{Q}=\{\mathcal{Q}_\theta^{(i)},\mathcal{S}|i\in[N], \theta\in[M]\}$.
    Then \begin{equation}\label{eq0}
      I(\mathcal{C}^{[N]};\mathcal{Q})=0,\end{equation}
      which implies that the user generates  queries without knowledge of the exact content of the coded records.
  \item {\bf Response phase.} For $1\leq i\leq N$, the $i$th $\rm{Serv}^{(i)}$ at receiving $\mathcal{Q}_\theta^{(i)}$, computes ${\rm Ans}^{(i)}(\mathcal{Q}_\theta^{(i)},\mathcal{C}^{(i)})=\mathcal{A}_\theta^{(i)}$ and sends it to the user, where ${\rm Ans}^{(i)}(\cdot,\cdot)$ is ${\rm Serv}^{(i)}$'s {\it answer function} determined by the scheme. ObviouslyEvidently, \begin{equation}\label{eq1}
      H(\mathcal{A}_\theta^{(i)}|\mathcal{C}^{(i)};\mathcal{Q}_\theta^{(i)})=0.\end{equation}
\end{itemize}

Moreover, a coded PIR scheme must satisfy the following two conditions:
\begin{itemize}
\item[(1)]{\it Correctness: }
     \begin{equation}\label{eqquery}
      H(\mathcal{W}_\theta|\mathcal{A}^{[N]}_\theta,\mathcal{Q}^{[N]}_\theta,\mathcal{S})=0,
      \end{equation}
      which implies that the user can definitely recover the record $\mathcal{W}_\theta$
      after receiving  responses from all servers. Based on the definition of $\mathcal{Q}$, the correctness conditions can also be represented as
      \begin{equation}\label{eq2}H(\mathcal{W}_\theta|\mathcal{A}^{[N]}_\theta,\mathcal{Q})=0.\end{equation}
\item[(2)]{\it Privacy:} For any $i\in[N]$,
\begin{equation}\label{eq3}
     I(\theta;\mathcal{Q}^{(i)}_\theta,\mathcal{A}^{(i)}_\theta,\mathcal{C}^{(i)})=0,
      \end{equation} which implies that any individual server gets  no information about the index $\theta$. Note that $\mathcal{Q}^{(i)}_\theta,\mathcal{A}^{(i)}_\theta,\mathcal{C}^{(i)}$ is the information held by the  $\rm{Serv}^{(i)}$.
\end{itemize}

Set $D=\sum_{i=1}^{N}H(\mathcal{A}^{(i)}_\theta)$, which actually denotes the download size. From the privacy condition, we have $I(\theta;\mathcal{A}^{(i)}_\theta)=0$, which implies that $D$ is independent of the index $\theta$.
Thus, we can define the rate and the capacity of PIR schemes as follows.

\begin{definition}\label{def0}(PIR Rate and Capacity) The PIR rate $\mathcal{R}$ of a PIR scheme is defined as
{\small \begin{equation*}
 \mathcal{R}= \frac{L}{D}=\frac{H(\mathcal{W}_\theta)}{\sum^N_{i=1}H(\mathcal{A}^{(i)}_\theta)} {\rm ~for~any~} \theta\in[M].
\end{equation*}}
The capacity $\mathcal{C}_{\mbox{\tiny C-PIR}}$ is the supremum of $\mathcal{R}$ over all PIR schemes.
\end{definition}
\begin{definition}\label{def1}(Sub-packetization and Access Number)
Suppose $L=K\tilde{L}$ and each record is expressed as a $K\times \tilde{L}$ matrix over $\mathbb{F}_q$, i.e., $W_j\in\mathbb{F}_q^{K\times \tilde{L}}$ for $1\leq j\leq M$. Using the notations defined previously,
a coded PIR scheme is called linear  if for retrieving any record $W_\theta$, $\theta\in[M]$, the answers from each server are derived as linear combinations of the data stored in that server, i.e., for $1\leq i\leq N$,
 {\small\begin{equation}\label{LPIR}
 A_\theta^{(i)}=\sum_{j=1}^MC^{(i)}_jQ_{\theta,j}^{(i)}\in\mathbb{F}_q^{~\gamma_i}\;,
 \end{equation} }
where $Q_{\theta,j}^{(i)}$ is an $\tilde{L}\times\gamma_i$ matrix over $\mathbb{F}_q$, $1\leq j\leq M$.
We call  $L$ as the sub-packetization of the PIR scheme. Moreover, we define the access number $\omega$ as the maximum number of sub-packets accessed by all servers for retrieving any record, i.e.,
{\small \begin{equation}\label{WPIR}
 \omega=\max_{\theta\in[M]}\sum_{i=1}^N\sum_{j=1}^M\texttt{RN}(Q_{\theta,j}^{(i)}),\;
 \end{equation} }
 where $\texttt{RN}(Q_{\theta,j}^{(i)})$ denotes the number of nonzero rows  in  $Q_{\theta,j}^{(i)}$.
\end{definition}

\subsection{Capacity of coded PIR schemes}\label{sec2b}
Note that the capacity of coded PIR has been determined in ~\cite{Bana&Uluk16:CapacityPIRCoded}, i.e.,
$\mathcal{C}_{\mbox{\tiny C-PIR}}=(1+\frac{K}{N}+\frac{K^2}{N^2}+\dots+\frac{K^{M-1}}{N^{M-1}})^{-1}.$
We briefly restate some key lemmas during the derivation of this capacity, which will be used in later sections. Proofs of the following two lemmas can be found in \cite{Bana&Uluk16:CapacityPIRCoded}.

\begin{lemma}\label{lem2} For a coded PIR scheme, for any $\theta,\theta^\prime\in[M]$, any subset $\Lambda\subseteq[M]$ and $i\in[N]$ ,
\begin{equation}\label{eq5}H(\mathcal{A}^{(i)}_{\theta}| \mathcal{W}_\Lambda,\mathcal{Q})=H(\mathcal{A}^{(i)}_{\theta^\prime}| \mathcal{W}_\Lambda,\mathcal{Q})\;.\end{equation}
\end{lemma}

\begin{lemma}\label{lem3} For a coded PIR scheme, for any $\theta\in[M]$, any subset $\Lambda\subseteq[M]$ and $\Gamma\subseteq[N]$ with $|\Gamma|=K$,
\begin{equation}\label{eq6}H(\mathcal{A}^{\Gamma}_{\theta}| \mathcal{W}_\Lambda,\mathcal{Q})=\sum_{i\in\Gamma}H(\mathcal{A}^{(i)}_{\theta}| \mathcal{W}_\Lambda,\mathcal{Q})\;.\end{equation}
\end{lemma}
Note that from (\ref{eq5}) and (\ref{eq6}) we can immediately determine that for any $\theta,\theta^\prime\in[M]$, any subset $\Lambda\subseteq[M]$, and $\Gamma\subseteq[N]$ with $|\Gamma|=K$,
\begin{equation}\label{eqqq}H(\mathcal{A}^{\Gamma}_{\theta}| \mathcal{W}_\Lambda,\mathcal{Q})=H(\mathcal{A}^{\Gamma}_{\theta^\prime}|\mathcal{W}_\Lambda,\mathcal{Q})\;.\end{equation}
\begin{lemma}\label{lem4} For a coded PIR scheme, for any subset $\Lambda\subseteq[M]$, for any $\theta\in\Lambda$ and any $\theta^\prime\in [M]-\Lambda$,
$$H(\mathcal{A}^{[N]}_{\theta}| \mathcal{W}_\Lambda,\mathcal{Q})\geq \frac{KL}{N}+\frac{K}{N} H(\mathcal{A}^{[N]}_{\theta^\prime}| \mathcal{W}_\Lambda,\mathcal{W}_{\theta^\prime},\mathcal{Q})\;.$$
\end{lemma}
\begin{proof} Since
\begin{equation}\label{eq7}
 H(\mathcal{A}^{[N]}_{\theta}| \mathcal{W}_\Lambda,\mathcal{Q})\geq H(\mathcal{A}^{\Gamma}_{\theta}| \mathcal{W}_\Lambda,\mathcal{Q}), \end{equation} for any $\Gamma \subseteq[N]$ with  $|\Gamma|=K$, then
\begin{equation}\label{eq8}
\begin{split}
 H(\mathcal{A}^{[N]}_{\theta}| \mathcal{W}_\Lambda,\mathcal{Q})
 \geq& \frac{1}{\binom{N}{K}}\sum_{\Gamma: \Gamma \subseteq[N],|\Gamma|=K}H(\mathcal{A}^{\Gamma}_{\theta}| \mathcal{W}_\Lambda,\mathcal{Q})\\
 \stackrel{(a)}{=}&\frac{1}{\binom{N}{K}}\sum_{\Gamma: \Gamma \subseteq[N],|\Gamma|=K}H(\mathcal{A}^{\Gamma}_{\theta^\prime}|\mathcal{W}_\Lambda,\mathcal{Q}) \\
 \stackrel{(b)}{\geq}&\frac{K}{N}H(\mathcal{A}^{[N]}_{\theta^\prime}| \mathcal{W}_\Lambda,\mathcal{Q})\\
 =&\frac{K}{N}\big(H(\mathcal{A}^{[N]}_{\theta^\prime}, \mathcal{W}_{\theta^\prime}| \mathcal{W}_\Lambda,\mathcal{Q})-H(\mathcal{W}_{\theta^\prime}|\mathcal{A}^{[N]}_{\theta^\prime},\mathcal{W}_\Lambda,\mathcal{Q})\big)\notag\\
 \stackrel{(c)}{=}&\frac{K}{N}\big(H(\mathcal{W}_{\theta^\prime}| \mathcal{W}_\Lambda,\mathcal{Q})+ H(\mathcal{A}^{[N]}_{\theta^\prime}| \mathcal{W}_\Lambda,\mathcal{W}_{\theta^\prime},\mathcal{Q})\big) \notag\\
 \stackrel{(d)}{=}&\frac{KL}{N} +\frac{K}{N}H(\mathcal{A}^{[N]}_{\theta^\prime}| \mathcal{W}_\Lambda,\mathcal{W}_{\theta^\prime},\mathcal{Q}),
 \end{split}
 \end{equation}
 where  {\small$(a)$} follows from (\ref{eqqq}), the inequality {\small$(b)$} comes from the Han's inequality,  {\small$(c)$} is due to the fact that  $H(\mathcal{W}_{\theta^\prime}|\mathcal{A}^{[N]}_{\theta^\prime},\mathcal{W}_\Lambda,\mathcal{Q})=0$, and {\small$(d)$} comes from  the assumptions (\ref{q1}) and (\ref{eq0}).
\end{proof}

The next theorem characterizes the capacity of $[N,K]$ MDS coded PIR in the non-colluding case (i.e. $T=1$). The theorem has been proved in \cite{Bana&Uluk16:CapacityPIRCoded}. Here, we reprove the theorem  to derive some key equalities for later use.

\begin{theorem}\label{thm1} For coded PIR with $M$ records and $N$ coded servers, the capacity is
\begin{align*}
\mathcal{C}_{\mbox{\tiny C-PIR}}=(1+\frac{K}{N}+\frac{K^2}{N^2}+\dots+\frac{K^{M-1}}{N^{M-1}})^{-1}.
\end{align*}
\end{theorem}
\begin{proof} Based on the general capacity-achieving coded PIR scheme presented in \cite{Bana&Uluk16:CapacityPIRCoded}, it is sufficient to demonstrate that for all coded PIR schemes, the PIR rate is bounded by
$\mathcal{R}\leq (1+\frac{K}{N}+\frac{K^2}{N^2}+\dots+\frac{K^{M-1}}{N^{M-1}})^{-1}$.

 For any $\theta\in[M]$, we prove
 \begin{equation}\label{eq9}
 H(\mathcal{A}_\theta^{[N]}|\mathcal{Q})\geq \sum^{M-1}_{i=1}\frac{K^i}{N^i} L\;.
 \end{equation}

 First, we have  \begin{equation}\label{eq10}
\begin{split}
 L=H(\mathcal{W}_\theta)
 \stackrel{(a)}{=}H(\mathcal{W}_\theta|\mathcal{Q})-
 H(\mathcal{W}_\theta|\mathcal{A}_\theta^{[N]},\mathcal{Q})
 =H(\mathcal{A}_\theta^{[N]}|\mathcal{Q})-H(\mathcal{A}_\theta^{[N]}|\mathcal{W}_\theta,\mathcal{Q}), \end{split}
\end{equation}
  where {\small$(a)$} comes from (\ref{eq0}) and (\ref{eq2}). Then by Lemma \ref{lem4},
  $H(\mathcal{A}_\theta^{[N]}|\mathcal{W}_\theta,\mathcal{Q})\geq \frac{KL}{N} + \frac{K}{N}(H(\mathcal{A}^{[N]}_{\theta^\prime}| \mathcal{W}_{\theta},\mathcal{W}_{\theta^\prime},\mathcal{Q})$.
 By recursively using Lemma \ref{lem4}, we have
\begin{equation}\label{eq11}
\begin{split}
H(\mathcal{A}_\theta^{[N]}|\mathcal{W}_1,\mathcal{Q})&\geq \sum^{M-1}_{i=1} \frac{K^i}{N^i}L+
\frac{K^{M-1}}{N^{M-1}} H(\mathcal{A}_{\theta^{\prime\prime}}^{[N]}|\mathcal{W}_{[M]},\mathcal{Q})
\stackrel{(a)}{=}\sum^{M-1}_{i=1} \frac{K^i}{N^i} L,
\end{split}
\end{equation}
where {\small$(a)$} comes from (\ref{eq1}). Combining with (\ref{eq10}) and (\ref{eq11}), we immediately obtain (\ref{eq9}).

Finally, for any coded PIR scheme, we know that its rate
\begin{equation}\label{eq12}\mathcal{R}=\frac{H(\mathcal{W}_\theta)}{\sum^N_{i=1}H(\mathcal{A}_\theta^{(i)})}\leq \frac{L}{H(\mathcal{A}_\theta^{[N]})}\leq \frac{L}{H(\mathcal{A}_\theta^{[N]}|\mathcal{Q})}.\end{equation}
Combining with (\ref{eq9}),
 $\mathcal{R}\leq (1+\frac{K}{N}+\frac{K^2}{N^2}+\dots+\frac{K^{M-1}}{N^{M-1}})^{-1}$.
\end{proof}

\section{The Lower Bound on Sub-Packetization}\label{sec3}
In this section, we derive a lower bound on the sub-packetization for all linear capacity-achieving coded PIR schemes. Namely,
\begin{theorem}\label{thm2} Suppose $M\geq 2, N\!>\!K\!\geq\!1$. Then any linear capacity-achieving $[N,K]$ MDS coded PIR scheme has sub-packetization
$L\geq Kn^{M-1}$ where $d={\rm gcd}(N,K),n=N/d$.
\end{theorem}
In proving the lower bound, we derive some identities of capacity-achieving coded PIR schemes and then some properties of linear capacity-achieving coded PIR schemes in Section \ref{sec3a} and Section \ref{sec3b}, respectively. Finally, the proof of Theorem \ref{thm2} is presented in Section \ref{sec3c}.

\subsection{Some identities for capacity-achieving PIR schemes}\label{sec3a}

\begin{lemma}\label{lem5} Consider capacity-achieving coded PIR schemes. For any $\theta\in[M]$, denote $\bar{\theta}=[M]-\theta$, then
\begin{equation}\label{eq13}
H(\mathcal{A}_\theta^{[N]}|\mathcal{Q})=\sum_{i=1}^{N}H(\mathcal{A}_\theta^{(i)})=D\;,
\end{equation}
\begin{equation}\label{eq14}
H(\mathcal{A}_\theta^{[N]}|\mathcal{W}_\theta,\mathcal{Q})=D-L\;.
\end{equation}
\begin{equation}\label{qq}
H(\mathcal{A}_\theta^{[N]}|\mathcal{W}_{\bar{\theta}},\mathcal{Q})=L\;.
\end{equation}

\end{lemma}
\begin{proof}From (\ref{eq12})
$\mathcal{R}\!=\!L/D\!\leq\!L/H(\mathcal{A}_\theta^{[N]}|\mathcal{Q})\!\leq (1+\frac{K}{N}+\frac{K^2}{N^2}+\dots+\frac{K^{M-1}}{N^{M-1}})^{-1}
$ for any coded PIR scheme. In particular, for every capacity-achieving coded PIR scheme, $\mathcal{R}=(1+\frac{K}{N}+\frac{K^2}{N^2}+\dots+\frac{K^{M-1}}{N^{M-1}})^{-1}.$ Therefore, we have
$H(\mathcal{A}_\theta^{[N]}|\mathcal{Q})=\sum_{i=1}^{N}H(\mathcal{A}_\theta^{(i)})=D$ for capacity-achieving coded PIR schemes. Combining with (\ref{eq10}) and (\ref{eq13}), we further have
$H(\mathcal{A}_\theta^{[N]}|\mathcal{W}_\theta,\mathcal{Q})=H(\mathcal{A}_\theta^{[N]}|\mathcal{Q})-L=D-L.$
Finally, by the fact $H(\mathcal{A}_\theta^{[N]}|\mathcal{W}_{[M]},\mathcal{Q})=0$ and $H(\mathcal{W}_{\theta}|\mathcal{A}_\theta^{[N]},\mathcal{W}_{\bar{\theta}},\mathcal{Q})=0$, we have
$$H(\mathcal{A}_\theta^{[N]}|\mathcal{W}_{\bar{\theta}},\mathcal{Q})=I(\mathcal{A}_\theta^{[N]};\mathcal{W}_{\theta}|\mathcal{W}_{\bar{\theta}},\mathcal{Q})=H(\mathcal{W}_{\theta}|\mathcal{W}_{\bar{\theta}},\mathcal{Q})=L.
$$
\end{proof}

\begin{lemma}\label{lem6} Consider capacity-achieving coded PIR schemes. For any $\theta\in[M]$, any $ \Lambda\subseteq[M]$, and any $\Gamma\subseteq[N]$ with $|\Gamma|=K$,
 \begin{equation}\label{eq15}H(\mathcal{A}_\theta^{\Gamma}|\mathcal{W}_{\Lambda},\mathcal{Q})= \begin{cases}
 H(\mathcal{A}_\theta^{[N]}|\mathcal{W}_{\Lambda},\mathcal{Q}),& if~\theta\in \Lambda\\
  \frac{K}{N}H(A_\theta^{[N]}|\mathcal{W}_{\Lambda},\mathcal{Q}), & if ~\theta\notin\Lambda
\end{cases}\end{equation}
\end{lemma}
\begin{proof}
For capacity-achieving coded PIR schemes, (\ref{eq7}) and (\ref{eq8}) both hold with equalities, i.e., for any $\theta \in \Lambda$,
 \begin{equation}\label{eq16}
 H(\mathcal{A}_\theta^{[N]}|\mathcal{W}_{\Lambda},\mathcal{Q})=H(\mathcal{A}_\theta^{\Gamma}|\mathcal{W}_{\Lambda},\mathcal{Q})
 \end{equation}
 while for any $\theta \in \Lambda,~\theta^\prime \in[M]-\Lambda$,
$H(\mathcal{A}_\theta^{[N]}|\mathcal{W}_{\Lambda},\mathcal{Q})=
 \frac{K}{N}H(\mathcal{A}_{\theta^\prime}^{[N]}|\mathcal{W}_{\Lambda},\mathcal{Q}).$

Thus we are left to prove the Lemma for the case $\theta\not\in\Lambda$.
Arbitrarily choose $\theta^\prime \in\Lambda$, then for any $\theta\not\in\Lambda$,
$$\frac{K}{N}H(\mathcal{A}_\theta^{[N]}|\mathcal{W}_{\Lambda},\mathcal{Q})
=H(\mathcal{A}_{\theta^\prime}^{[N]}|\mathcal{W}_{\Lambda},\mathcal{Q})\stackrel{(a)}{=}H(\mathcal{A}_{\theta^\prime}^{\Gamma}|\mathcal{W}_{\Lambda},\mathcal{Q}) \stackrel{(b)}{=}H(\mathcal{A}_{\theta}^{\Gamma}|\mathcal{W}_{\Lambda},\mathcal{Q}),$$
where {\small$(a)$} come from (\ref{eq16}) and  {\small$(b)$} come from (\ref{eqqq}).
\end{proof}

\subsection{Properties of linear capacity-achieving coded PIR schemes}\label{sec3b}
We first define a vectorization operator ${\rm Vec}$, which maps a matrix $A\in\mathbb{F}^{m\times n}$ to a row vector ${\rm Vec}(A)\in\mathbb{F}^{mn}$ whose entries are successively drawn from the matrix row by row. For example, suppose {\small $A=\begin{pmatrix}1&2&0\\2&0&1\end{pmatrix}$,} then {\small ${\rm Vec}(A)=\begin{pmatrix}1&2&0&2&0&1\end{pmatrix}$.}
 The proof of the following proposition is not difficult and so we omit it here.
\begin{proposition}\label{pro1} Suppose $A\in\mathbb{F}^{m\times s},B\in\mathbb{F}^{t\times n},Z\in\mathbb{F}^{s\times t},Y\in\mathbb{F}^{m\times n}$ and $AZB=Y$, then
$
{\rm Vec}(Y)={\rm Vec}(Z)(A^{\tau}\otimes B).
$ Moreover, suppose $A_1,A_2\in\mathbb{F}^{m\times n}$, $k_1,k_2\in\mathbb{F}$, then
$
{\rm Vec}(k_1A_1+k_2A_2)=k_1{\rm Vec}(A_1)+ k_2{\rm Vec}(A_2).
$
\end{proposition}

By Proposition \ref{pro1}, we can rewrite the equation (\ref{LPIR}), i.e,
$A^{(i)}_\theta={\rm Vec}(A^{(i)}_\theta)=\sum^M_{j=1}{\rm Vec}(W_j)({\bf g}_i\otimes Q^{(i)}_{\theta,j}).$
Equivalently,
{\small\begin{eqnarray} \label{LPIR1}
  &&(A_\theta^{(1)},A_\theta^{(2)},...,A_\theta^{(N)})
  =({\rm Vec}(W_1),{\rm Vec}(W_2),...,{\rm Vec}(W_M))\cdot
  \begin{pmatrix}
  {\bf g}_1\otimes Q_{\theta,1}^{(1)}&{\bf g}_2 \otimes Q_{\theta,1}^{(2)}&\cdots&{\bf g}_N\otimes Q_{\theta,1}^{(N)}\\
  {\bf g}_1\otimes Q_{\theta,2}^{(1)}&{\bf g}_2\otimes Q_{\theta,2}^{(2)}&\cdots&{\bf g}_N\otimes Q_{\theta,2}^{(N)}\\
  \vdots&\vdots&\vdots&\vdots\\
  {\bf g}_1\otimes Q_{\theta,M}^{(1)}&{\bf g}_2\otimes Q_{\theta,M}^{(2)}&\cdots&{\bf g}_N\otimes Q_{\theta,M}^{(N)}
  \end{pmatrix}\;.
\end{eqnarray}}
That is, we formulate of a general linear coded PIR scheme in (\ref{LPIR1}). In particular,  we represent each record as a row vector by using the function {\rm Vec}, which is convenient for investigating the rank of the matrix.
For any $\Gamma\subseteq[N],\Lambda\subseteq[M]$, define the sub-matrix $\tilde{ Q}^{\Gamma}_{\theta,\Lambda}=({\bf g}_i\otimes Q^{(i)}_{\theta,j})_{j\in\Lambda,i\in\Gamma}.$

Next, we establish a connection between the rank of the sub-matrix $\tilde{ Q}^{\Gamma}_{\theta,\Lambda}$ and some conditional entropy in Lemma \ref{lem7}. Combining with the identities of the entropy obtained in Section III-A, we can then get some characterizations of these sub-matrices in Proposition \ref{pro2}, which will be used to prove the lower bound on sub-packetization in Section III-C.

\begin{lemma}\label{lem7}Consider linear capacity-achieving coded PIR scheme. For any $\theta\in[M]$, and for any nonempty subsets $\Gamma\subseteq[N]$ and $\Lambda\subseteq[M]$,
\begin{equation}\label{q2}
H(\mathcal{A}_\theta^{\Gamma}|\mathcal{W}_{\Lambda},\mathcal{Q})={\rm rank}(\tilde{ Q}^{\Gamma}_{\theta,[M]-\Lambda}).
\end{equation}
\end{lemma}
The proof of this lemma is similar to that of Lemma 8 in \cite{Zhang&Xu17:OptimalSubpacketization} and is omitted here.
\begin{proposition}\label{pro2} For a linear capacity-achieving coded PIR scheme, for any $\theta\in[M]$,
then it holds
\begin{equation}\label{q0}
{\rm rank}(\tilde{Q}^{[N]}_{\theta,\theta})=L\;.
\end{equation}
Moreover, for any $\Gamma\subseteq[N]$ with $|\Gamma|=K$,
it holds
\begin{equation}\label{q3}
{\rm rank}(\tilde{ Q}^{\Gamma}_{\theta,\theta})=\frac{KL}{N},
\end{equation}
\begin{equation}\label{q4}
{\rm rank}(\tilde{ Q}^{\Gamma}_{\theta,\bar{\theta}})=D-L.
\end{equation}
\end{proposition}
\begin{proof} First, it follows from (\ref{q2}) and (\ref{qq}) that ${\rm rank}(\tilde{ Q}^{[N]}_{\theta,\theta})=H(\mathcal{A}_\theta^{[N]}|\mathcal{W}_{\bar{\theta}},\mathcal{Q})=L$. Hence, by Lemma \ref{lem6}, it has ${\rm rank}(\tilde{ Q}^{\Gamma}_{\theta,\theta})=\frac{KL}{N}$. Similarly, one can obtain
${\rm rank}(\tilde{ Q}^{\Gamma}_{\theta,\bar{\theta}})={\rm rank}(\tilde{ Q}^{[N]}_{\theta,\bar{\theta}})=D-L$.
\end{proof}

\subsection{Proof of Theorem \ref{thm2}}\label{sec3c}

We first present a simple lemma without proof.

\begin{lemma}\label{lem8} Let $a,b,m\in\mathbb{N}$. Suppose $d_1={\rm gcd}(a,b),d_2={\rm gcd}(a^m,\sum^{m}_{i=0}a^{m-i}b^i)$, then $d_2=d_1^m$.
\end{lemma}

Next we prove Theorem \ref{thm2}.

\begin{proof} The proof is completed in four steps.

 {\it (1) Prove $L$ and $D$ have specific forms},
that is,
$L=\mu \frac{dk}{{\rm gcd}(d,n^{M-1})}n^{M-1}$ and $D=\mu\frac{dk}{{\rm gcd}(d,n^{M-1})}\sum^{M-1}_{i=0}n^{M-1-i}k^i$ for some $\mu\in\mathbb{N}$,
where $d={\rm gcd}(N,K), n=\frac{N}{d}, k=\frac{K}{d}$.

 By the definition of linear capacity-achieving coded PIR schemes, we have
\begin{align}\label{eqrate}
\frac{L}{D}=\frac{1}{1+\frac{k}{n}+\frac{k^2}{n^2}+\dots+\frac{k^{M-1}}{n^{M-1}}}
=\frac{n^{M-1}}{\sum^{M-1}_{i=0}n^{M-1-i}k^i}.
\end{align}
Since  both $L$ and $D$ are integers in linear schemes, then (\ref{eqrate}) implies $n^{M-1}|L\sum^{M-1}_{i=0}n^{M-1-i}k^i$.
Note that ${\rm gcd}(n,k)=1$, and so by Lemma \ref{lem8} it holds
 $${\rm gcd}(n^{M-1},\sum^{M-1}_{i=0}n^{M-1-i}k^i)=({\rm gcd}(n,k))^{M-1}=1.$$
 Since  $n^{M-1}|L(\sum^{M-1}_{i=0}n^{M-1-i}k^i)$, then $n^{M-1}$ is a factor of $L$. On the other hand, by the assumption of $L$ i.e., $L=K\tilde{L}$, which implies that $L$ is a multiple of ${\rm lcm}(K,n^{M-1})$. Denote $\mu=\frac{L}{{\rm lcm}(K,n^{M-1})}$,  then
 \begin{equation}
 L=\mu \cdot{\rm lcm}(K,n^{M-1})=\mu \frac{dk}{{\rm gcd}(d,n^{M-1})}n^{M-1}. \notag
 \end{equation}
 Combining with  (\ref{eqrate}),
 \begin{equation*}
 D=\mu \frac{dk}{{\rm gcd}(d,n^{M-1})}\sum^{M-1}_{i=0}n^{M-1-i}k^i.
\end{equation*}
\vspace{6pt}

{\it (2) Prove $N\mid L$}.

By Lemma \ref{lem3} and Lemma \ref{lem7}, we have ${\rm rank}(\tilde{ Q}^{\Gamma}_{\theta,\theta})=\sum_{i\in\Gamma}{\rm rank}(\tilde{ Q}^{(i)}_{\theta,\theta})$ for any $\Gamma$ with $|\Gamma|=K$. So for any $i,j\in[N]$, it holds ${\rm rank}(\tilde{Q}^{(i)}_{\theta,\theta})={\rm rank}(\tilde{ Q}^{(j)}_{\theta,\theta})$.  Combining with (\ref{q3}), one can  obtain ${\rm rank}(\tilde{ Q}^{(i)}_{\theta,\theta})=\frac{L}{N}$, which implies $N|L$.
\vspace{6pt}

{\it (3) Prove $K\mid D-L$}.

Similar to the above result, for any $i,j\in[N]$, we have ${\rm rank}(\tilde{ Q}^{(i)}_{\theta,\bar{\theta}})={\rm rank}(\tilde{ Q}^{(j)}_{\theta,\bar{\theta}})$. Combining with (\ref{q4}), one can obtain ${\rm rank}(\tilde{ Q}^{(i)}_{\theta,\theta})=\frac{D-L}{K}$, which implies $K|D-L$.
\vspace{6pt}

{\it (4) Prove ${\rm gcd}(d,n^{M-1})|\mu$}.

Finally, from $N\mid L$ and $L=\mu k\frac{d}{{\rm gcd}(d,n^{M-1})}n^{M-1}$  we have ${\rm gcd}(d,n^{M-1})\mid \mu k n^{M-2}$. Similarly, from $K\mid D-L$ and $D-L=\mu k\frac{dk}{{\rm gcd}(d,n^{M-1})} \sum^{M-2}_{i=0}n^{M-2-i}k^i$ we have ${\rm gcd}(d,n^{M-1})\mid \mu k \sum^{M-2}_{i=0}n^{M-2-i}k^i$. Therefore,
$${\rm gcd}(d,n^{M-1})\mid{\rm gcd}(\mu k n^{M-2},\mu k\sum^{M-2}_{i=0}n^{M-2-i}k^i)\;.$$
Note that ${\rm gcd}(n,k)=1$,  then we know from Lemma \ref{lem8} that ${\rm gcd}(n^{M-2},\sum^{M-2}_{i=0}n^{M-2-i}k^i)=1$, which implies that ${\rm gcd}(\mu k n^{M-2},\mu k\sum^{M-2}_{i=0}n^{M-2-i}k^i)=\mu k$. On the other hand, ${\rm gcd}(k,{\rm gcd}(d,n^{M-1}))={\rm gcd}(k,d,n^{M-1})=1$. So we have ${\rm gcd}(d,n^{M-1})|\mu$ and $\mu \geq {\rm gcd}(d,n^{M-1})$. Consequently, $$L=\mu k\frac{d}{{\rm gcd}(d,n^{M-1})}n^{M-1}\geq K n^{M-1}\;.$$
\end{proof}

\section{The Lower Bound on Access Number}\label{sec4}
In this section, we derive a lower bound on the access number for all linear capacity-achieving coded PIR schemes. Namely,
\begin{theorem}\label{thm3} For all linear capacity-achieving coded PIR schemes with $M$ records stored in $N$ servers by using an $[N,K]$ {\rm MDS} code, the access number $\omega$ is bounded by
$\omega\geq MKn^{M-1}$, where $n=\frac{N}{{\rm gcd}(N,K)},M\geq 2,N>K$.
\end{theorem}
\begin{proof}  Note that in the second part of the proof of Theorem \ref{thm2}, it holds ${\rm rank}(Q_{j,j}^{(i)})=\frac{L}{N}$ for $i\in[N],j\in[M]$. Combining with (\ref{eq5}) and (\ref{q2}), one can obtain ${\rm rank}(\tilde{Q}_{\theta,j}^{(i)})=\frac{L}{N}$ for $i\in[N],j\in[M]$. On the other hand,
since  $|\texttt{RN}(Q_{\theta,j}^{(i)})|\geq {\rm rank}(Q_{\theta,j}^{(i)})$ and  $\tilde{Q}_{\theta,j}^{(i)}={\bf g_i}\otimes Q_{\theta,j}^{(i)}$, then
 $\omega\geq\sum_{i=1}^N\sum_{j=1}^M {\rm rank}(\tilde{Q}_{\theta,j}^{(i)})=NM\frac{L}{N}=ML\stackrel{(a)}{\geq} MKn^{M-1},$
where {\small $(a)$} follows from  Theorem \ref{thm2}.
\end{proof}

\section{General coded PIR schemes with $L=Kn^{M-1}$ and $\omega=MKn^{M-1}$ }\label{sec5}
In this section, we present a linear capacity-achieving  coded PIR scheme with the sub-packetization $L=Kn^{M-1}$  and the  access number $\omega=MKn^{M-1}$ for nontrivial cases, i.e., $N>K\geq1, M>1$. To illustrate the main idea, we begin with three examples. The first two examples are for the case $K<N<2K$ and the third is for the case $N\geq 2K$.
 \subsection{Examples}\label{sec5a}
\begin{example}\label{eg1} Suppose $M=2,N=3$ and $K=2$. In this case the sub-packetization of our scheme is $L=Kn^{M-1}=6$, and so each record can be regarded as a $2\times 3$ matrix over $\mathbb{F}_q$, i.e., $W_1,W_2\in\mathbb{F}^{2\times 3}_q$. Let ${\bf g}_i$ be the $i$th column of a $2\times 3$ generator matrix $G$ of an $[3,2]$ MDS code over $\mathbb{F}_q$, which is used for distributed storage. That is, the data stored in  ${\rm Serv}^{(i)}$, $1\leq i\leq 3$, is ${\bf g}^\tau_i(W_1,W_2)$. Without loss of generality, suppose the the user wants $W_1$. The PIR scheme works as follows:

First, let $S_1,S_2$ be two matrices privately chosen by the user independently and uniformly from all $3\times3$  permutation matrices, where a permutation matrix is a binary matrix  with only one $1$ in each row and each column. Define
$({\bf a}_1,{\bf a}_2,{\bf a}_3)=W_1S_1,~({\bf b}_1,{\bf b}_2,{\bf b}_3)=W_2S_2,$
where ${\bf a}_i,{\bf b}_i$ are $2$-dimensional column vectors for  $1\leq i\leq 3$.
Since the ${\bf a}_i$'s contain information of the desired record $W_1$, we call them {\it desired columns}, while the ${\bf b}_i$'s and ${\bf a}_i+{\bf b}_j$'s are called {\it interference columns} and {\it mixed columns}, respectively.

However, because of the distributed storage using an $[3,2]$ MDS code, ${\rm Serv}^{(i)}$ can provide ${\bf g}^\tau_i{\bf a}_j,{\bf g}^\tau_i{\bf b}_j$ for $1\leq j\leq 3$ in his answers. We display all the answers in the left table in Figure~\ref{fg1}. Specifically, these answers are formed by iteratively applying the following two steps:
\begin{itemize}
 \item[(a1)]\emph{Combining new desired columns with recoverable interference columns.}
 \item[(a2)]\emph{Querying new interference columns to enforce record symmetry within each server.}
\end{itemize}
As in the left table in Figure~~\ref{fg1}, the second line is built based on the first line by using (a2). Thus, the record symmetry is achieved within each server in the first two lines. For example, ${\rm Serv}^{(1)}$ provides two symbols related with each record, while
${\rm Serv}^{(2)}$ and ${\rm Serv}^{(3)}$ each provides one symbol related with each record. After the two lines, one can see that the columns ${\bf a}_1, {\bf a}_2, {\bf b}_1, {\bf b}_2$ are recoverable from the $[3,2]$ MDS encoding. Then, using (a1), the third line of the table is formed.

\begin{figure}[ht]
\centering
\begin{tikzpicture}[scale=2]
\node at (0,0){\footnotesize
\begin{tabular}{ccc}
\specialrule{0.09em}{0pt}{1.2pt}
$\rm{Serv}^{(1)}$ & $\rm{Serv}^{(2)}$ & $\rm{Serv}^{(3)}$\\\specialrule{0.06em}{1.2pt}{1.2pt}
${\bf g}^\tau_1{\bf a}_1,{\bf g}^\tau_1{\bf a}_2$&${\bf g}^\tau_2{\bf a}_1$&${\bf g}^\tau_3{\bf a}_2$\\\specialrule{0em}{1.2pt}{1.2pt}
${\bf g}^\tau_1{\bf b}_1,{\bf g}^\tau_1{\bf b}_2$&${\bf g}^\tau_2{\bf b}_1$&${\bf g}^\tau_3{\bf b}_2$\\\specialrule{0em}{1.2pt}{1.2pt}
&${\bf g}^\tau_2({\bf a}_3+{\bf b}_2)$&${\bf g}^\tau_3({\bf a}_3+{\bf b}_1)$\\ \specialrule{0.09em}{0.5pt}{0pt}
\end{tabular}};

\draw[dotted](-1.48,-0.02)--(1.48,-0.02);
\draw[dotted](-1.48,-0.28)--(1.48,-0.28);

\draw(-1.49,0.53)--(-1.49,-0.53);
\draw(-0.52,0.53)--(-0.52,-0.53);
\draw(0.49,0.53)--(0.49,-0.53);
\draw(1.49,0.53)--(1.49,-0.53);

\draw[->] (1.5,0.12) arc (55:-55:0.15);
\draw[->] (1.5,-0.18) arc (55:-55:0.15);
\node at (1.7,0){\scriptsize(a2)};
\node at (1.7,-0.3){\scriptsize(a1)};

\node at (4,0){\footnotesize
\begin{tabular}{ccc}
\specialrule{0.09em}{0pt}{1.2pt}
$\rm{Serv}^{(1)}$ & $\rm{Serv}^{(2)}$ & $\rm{Serv}^{(3)}$\\\specialrule{0.06em}{1.2pt}{1.2pt}
${\bf g}^\tau_1{\bf a}'_1,{\bf g}^\tau_1{\bf a}'_2$&${\bf g}^\tau_2{\bf a}'_1$&${\bf g}^\tau_3{\bf a}'_2$\\\specialrule{0em}{1.2pt}{1.2pt}
${\bf g}^\tau_1{\bf b}'_1,{\bf g}^\tau_1{\bf b}'_2$&${\bf g}^\tau_2{\bf b}'_1$&${\bf g}^\tau_3{\bf b}'_2$\\\specialrule{0em}{1.2pt}{1.2pt}
&${\bf g}^\tau_2({\bf a}'_2+{\bf b}'_3)$&${\bf g}^\tau_3({\bf a}'_1+{\bf b}'_3)$\\ \specialrule{0.09em}{0.5pt}{0pt}
\end{tabular}};

\draw[dotted](2.52,-0.02)--(5.48,-0.02);
\draw[dotted](2.52,-0.28)--(5.48,-0.28);

\draw(2.51,0.53)--(2.51,-0.53);
\draw(3.48,0.53)--(3.48,-0.53);
\draw(4.49,0.53)--(4.49,-0.53);
\draw(5.49,0.53)--(5.49,-0.53);

\end{tikzpicture}
\caption{Suppose $M\!=\!2,N\!=\!3,K\!=\!2$. The left table is for privately retrieving $W_1$, while the right is for retrieving $W_2$.}
\label{fg1}
\end{figure}

It is clear that the user can recover all the columns ${\bf a}_1, {\bf a}_2, {\bf a}_3$ from all the answers listed in the table. By multiplying $S_1^{-1}$  the user can then obtain the record $W_1$. Thus,  the correctness condition is satisfied. We then explain why the privacy condition also holds. It is equivalent to show that for any  individual server, its query sequence for retrieving $W_1$ has the same distribution as the query sequence for retrieving $W_2$. The answers for retrieving $W_2$ are listed in the right table in Figure~\ref{fg1}, where the columns ${\bf a}'_i,{\bf b}'_i$ are defined as
$({\bf a}'_1,{\bf a}'_2,{\bf a}'_3)=W_1S'_1,~~({\bf b}'_1,{\bf b}'_2,{\bf b}'_3)=W_2S'_2.$
and the matrices $S_1'$ and $S_2'$ are random permutation matrices.

For any individual server, say ${\rm Serv}^{(2)}$, and any random permutation matrices $S_1,S_2$, we show that there exist corresponding choices of $S_1',S_2'$ such that the answers
of ${\rm Serv}^{(2)}$ remain the same in  both tables in Figure~\ref{fg1}. Specifically, set
$S_1'=({\bf s}_{1,1},{\bf s}_{1,3},{\bf s}_{1,2}),~S_2'=({\bf s}_{2,1},{\bf s}_{2,3},{\bf s}_{2,2}),$
where ${\bf s}_{i,j}$ denotes the $j$th column of $S_i$ for $i=1,2$ and $j=1,2,3$.
Then $({\bf a}_1,{\bf b}_1,{\bf a}_3+{\bf b}_2)=({\bf a}'_1,{\bf b}'_1,{\bf a}'_2+{\bf b}'_3)$, which implies that
the query sequence for retrieving $W_1$ has the same distribution as the query sequence for retrieving $W_2$ for the server ${\rm Serv}^{(2)}$. Similarly, one can find the corresponding permutation matrices $S'_1,S'_2$ for any individual server. Since the permutation matrices are randomly chosen and privately known by the user, the privacy condition is satisfied.

The total number of downloaded symbols from all servers is $4+3+3=10$ and each record consists of $6$ symbols. Hence, the PIR rate is $\frac{6}{10}=\frac{3}{5}$, which matches the capacity for this case. To compute the access number, one first notes that ${\bf g}^\tau_j{\bf a}_i={\bf g}^\tau_jW_1{\bf s}_{1,i}$ where ${\bf g}^\tau_jW_1$ is the partial data stored in ${\rm Serv}^{(j)}$ and ${\bf s}_{1,i}$ is a binary vector of weight $1$, and so providing ${\bf g}^\tau_j{\bf a}_i$  requires accessing only one sub-packet stored in ${\rm Serv}^{(j)}$. Consequently, the access number of the tables in Figure~\ref{fg1} is exactly the number of ${\bf g}^\tau_j{\bf a}_i, {\bf g}^\tau_{j}{\bf b}_i$'s involved in the tables, which is $12$ and attains the lower bound  of the access number.
\end{example}

\begin{example}\label{eg2}
Suppose $M\!=\!3,N\!=\!3,K\!=\!2$. Then the sub-packetization of our scheme is $L\!=\!Kn^{M-1}\!=\!18$, and so the records are denoted by $W_1,W_2,W_3\in\mathbb{F}_{q}^{2\times9}$. Similar to the Example \ref{eg1}, the data ${\bf g}^\tau_i(W_1,...,W_M)$ is stored in the ${\rm Serv}^{(i)}$ for all $i\in[3]$, where ${\bf g}_i$ is the $i$th column of $G\in\mathbb{F}_q^{2\times 3}$,  which is a generator matrix of an $[3,2] ~{\rm MDS}$ storage code. WLOG, suppose the desired record is $W_1$.

First, define
${\bf a}_{[1:9]}=W_1S_1,~{\bf b}_{[1:9]}=W_2S_2,~{\bf c}_{[1:9]}=W_3S_3,$
where $S_1,S_2,S_3$ are privately chosen by the user independently and uniformly from all $9\times9$ permutation matrices.
The user then iteratively applies Step (a1) and (a2) to generate the queries for all servers, which is displayed in Figure~\ref{fg3}. For simplicity, from now on we use the notation $\underline{\bf a}_i$ to denote ${\bf g}^\tau_j{\bf a}_i$ if it appears in ${\rm Serv}^{(j)}$'s answers. The same simplification on notations are induced for $\underline{\bf b}_i$'s and $\underline{\bf c}_i$'s.
\begin{figure}[ht]
\centering
\begin{tikzpicture}[scale=2]
\node at (0,0){\footnotesize\begin{tabular}{ccc}
\specialrule{0.09em}{0pt}{1.5pt}$\rm{Serv}^{(1)}$ & $\rm{Serv}^{(2)}$ & $\rm{Serv}^{(3)}$\\\specialrule{0.09em}{1pt}{1pt}
$\underline{\bf a}_1,\underline{\bf a}_2$&$\underline{\bf a}_1,\underline{\bf a}_3,\underline{\bf a}_4$&$\underline{\bf a}_2,\underline{\bf a}_3,\underline{\bf a}_4$\\\specialrule{0em}{1.2pt}{1.2pt}
$\underline{{\bf b}}_1,\underline{{\bf b}}_2$&$\underline{{\bf b}}_1,\underline{{\bf b}}_3,\underline{{\bf b}}_4$&$\underline{{\bf b}}_2,\underline{{\bf b}}_3,\underline{{\bf b}}_4$\\
$\underline{{\bf c}}_1,\underline{{\bf c}}_2$&$\underline{{\bf c}}_1,\underline{{\bf c}}_3,\underline{{\bf c}}_4$&$\underline{{\bf c}}_2,\underline{{\bf c}}_3,\underline{{\bf c}}_4$\\\specialrule{0em}{1.2pt}{1.2pt}
$\underline{\bf a}_5+\underline{\bf b}_3$&$\underline{\bf a}_5+\underline{\bf b}_2$&\\
$\underline{\bf a}_6+\underline{\bf b}_4$&&$\underline{\bf a}_6+\underline{\bf b}_1$\\
$\underline{\bf a}_7+\underline{\bf c}_3$&$\underline{\bf a}_7+\underline{\bf c}_2$&\\
$\underline{\bf a}_8+\underline{\bf c}_4$&&$\underline{\bf a}_8+\underline{\bf c}_1$\\\specialrule{0em}{1.2pt}{1.2pt}
$\underline{\bf b}_5+\underline{\bf c}_5$&$\underline{\bf b}_5+\underline{\bf c}_5$&\\
$\underline{\bf b}_6+\underline{\bf c}_6$&&$\underline{\bf b}_6+\underline{\bf c}_6$\\\specialrule{0em}{1.2pt}{1.2pt}
 &$\underline{\bf a}_9+\underline{\bf b}_6+\underline{\bf c}_6$&$\underline{\bf a}_9+\underline{\bf b}_5+\underline{\bf c}_5$\\
\specialrule{0.09em}{0.5pt}{0pt}
\end{tabular}};

\draw[dotted](-1.4,0.78)--(1.4,0.78);
\draw[dotted](-1.4,0.31)--(1.4,0.31);
\draw[dotted](-1.4,-0.62)--(1.4,-0.62);
\draw[dotted](-1.4,-1.12)--(1.4,-1.12);

\draw(-1.41,-1.35)--(-1.41,1.34);
\draw(-0.7,-1.35)--(-0.7,1.34);
\draw(0.34,-1.35)--(0.34,1.34);
\draw(1.41,-1.35)--(1.41,1.34);

\draw[->] (1.42,0.92) arc (55:-55:0.22);
\draw[->] (1.42,0.42) arc (55:-55:0.22);
\draw[->] (1.42,-0.44) arc (55:-55:0.22);
\draw[->] (1.42,-0.94) arc (55:-55:0.22);
\node at (1.65,0.7){\scriptsize(a2)};
\node at (1.65,0.27){\scriptsize(a1)};
\node at (1.65,-0.64){\scriptsize(a2)};
\node at (1.65,-1.12){\scriptsize(a1)};
\end{tikzpicture}
\caption{Query sequence for $\theta=1$ in the $(M=3,N=3,K=2)$ PIR scheme.}
\label{fg3}
\end{figure}

The correctness condition and the privacy condition can be similarly verified as in Example~\ref{eg1}. The access number of this scheme is $54$, and the PIR rate of this scheme is $\frac{18}{38}=\frac{9}{19}$, respectively, which achieves the lower bounds for this case.
\end{example}

\begin{example}\label{eg3} Suppose $M=2,N=5$ and $K=2$. Then the sub-packetization is $L=Kn^{M-1}=10$, and so each record is regarded as  a $2\times 5$ matrix over $\mathbb{F}_q$, i.e., $W_1,W_2\in\mathbb{F}^{2\times 5}_q$. Let ${\bf g}_i$ for $i\in[5]$ be the $i$th column of a matrix $G\in\mathbb{F}^{2\times 5}_q$, which is a generator matrix of an $[5,2]$ MDS code used for distributed storage. Then the data stored in ${\rm Serv}^{(i)}$ is ${\bf g}^\tau_i(W_1,W_2)$. WLOG, assume the desired record is $W_1$.

Let $S_1,S_2$ be two binary  matrices privately chosen by the user independently and uniformly from all $5\times5$ permutation matrices. Then define
${\bf a}_{[1:5]}=W_1S_1,~{\bf b}_{[1:5]}=W_2S_2$.
Then the queries to all servers are displayed in Figure~\ref{fg4}.
\begin{figure}[ht]
\centering
\begin{tikzpicture}[scale=2]
\node at (0,0){\footnotesize\begin{tabular}{ccccc}
\specialrule{0.09em}{0pt}{1.5pt}
$\rm{Serv}^{(1)}$ & $\rm{Serv}^{(2)}$ & $\rm{Serv}^{(3)}$&$\rm{Serv}^{(4)}$&$\rm{Serv}^{(5)}$\\\specialrule{0.09em}{1pt}{1.5pt}
&&&$\underline{\bf a}_1,\underline{\bf a}_2$&$\underline{\bf a}_1,\underline{\bf a}_2$\\\specialrule{0em}{1.2pt}{1.2pt}
&&&$\underline{\bf b}_1,\underline{\bf b}_2$&$\underline{\bf b}_1,\underline{\bf b}_2$ \\\specialrule{0em}{1.2pt}{1pt}
$\underline{\bf a}_3+\underline{\bf b}_1$&$\underline{\bf a}_3+\underline{\bf b}_1$&$\underline{\bf a}_4+\underline{\bf b}_1$&&\\
$\underline{\bf a}_4+\underline{\bf b}_2$&$\underline{\bf a}_5+\underline{\bf b}_2$&$\underline{\bf a}_5+\underline{\bf b}_2$&&\\
\specialrule{0.09em}{0.5pt}{0pt}
\end{tabular}};
\draw[dotted](-1.74,0.07)--(1.74,0.07);
\draw[dotted](-1.74,-0.18)--(1.74,-0.18);
\draw(-1.75,0.65)--(-1.75,-0.65);
\draw(-1.04,0.65)--(-1.04,-0.65);
\draw(-0.34,0.65)--(-0.34,-0.65);
\draw(0.4,0.65)--(0.4,-0.65);
\draw(1.1,0.65)--(1.1,-0.65);
\draw(1.75,0.65)--(1.75,-0.65);
\draw[->] (1.76,0.23) arc (55:-55:0.15);
\draw[->] (1.76,-0.12) arc (55:-55:0.15);
\node at (1.96,0.09){\scriptsize(a2)};
\node at (1.96,-0.25){\scriptsize(a1)};

\end{tikzpicture}
\caption{Query sequence for the $\theta=1$ in the $(M=2,N=5,K=2)$ PIR scheme.}
\label{fg4}
\end{figure}

The correctness condition and the privacy condition can be similarly verified as in Example~\ref{eg1}. The access number of this scheme is $20$, and the PIR rate of this scheme is $\frac{10}{14}=\frac{5}{7}$, respectively, which attains the lower bounds for this case.
\end{example}

\subsection{Formal description of the general scheme}\label{sec5b}
In addition to the notations defined in Section II, a formal description of our general scheme requires some additional notations that are listed in Table~\ref{tabn} for a quick check. As in the examples, the user first privately selects binary matrices $S_1,...,S_M$ independently and uniformly from all $\tilde{L}\times \tilde{L}$ binary permutation matrices. Define
$U_i=W_iS_i=({\bf u}_{i,1},...,{\bf u}_{i,\tilde{L}}),1\leq i\leq M$,
where ${\bf u}_{i,j}, 1\leq j\leq\tilde{L}$ is a $K$-dimensional column vector. For any subset $\Lambda\subseteq[M]$, we call ${\bf q}_{\Lambda,\lambda}=\sum_{i\in\Lambda}{\bf u}_{i,i_{\lambda}}$ a $\Lambda$-type $|\Lambda|$-sum, where $\lambda\in\mathbb{N}, i_{\lambda}\in[\tilde{L}]$. Evidently, ${\bf q}_{\Lambda,\lambda}$ is a desired column for $\Lambda=\{\theta\}$, an interference column for $\theta \notin\Lambda $, and a mixed column for $\{\theta\}\subsetneq\Lambda$.

\begin{table}[ht]
\centering
\caption{Notations in the General Scheme}\label{tabn}
\resizebox{0.8\textwidth}{!}{
{\footnotesize
\begin{tabular}{|c|l|}
  \hline
  $\theta$& {\rm the index of the desired record}\\\hline
  $\underline{\Lambda}$& {$\Lambda\cup\{\theta\}${\rm ~ for~some~}$\Lambda\subseteq[M]-\{\theta\}$}\\\hline
  $\Lambda_{j,h}$&{\rm suppose} $\{\Lambda\subseteq[M]-\{\theta\}\mid|\Lambda|=j\}=\{\Lambda_{j,1},\Lambda_{j,2},...,\Lambda_{j,t}\}${\rm ~where~}$1\leq h\leq t=\binom{M-1}{j}$\\\hline
  $q_{\Lambda,\lambda}$& {\rm a $\Lambda$-type $|\Lambda|$-sum}\\\hline
  $q^{(i)}_{\Lambda,h}$& {\rm the $h$th $\Lambda$-type $|\Lambda|$-sum provided by Serv$^{(i)}$}\\\hline
  $\gamma^{(i)}_j$ & {\rm the number of each type of $j$-sums provided by Serv$^{(i)}$}\\\hline
  $\alpha_j$& $\gamma^{(i)}_j$ for $ 1\leq i\leq N-K$\\\hline
  $\beta_j$&$\gamma^{(i)}_j$ for $N-K+1\leq i\leq N$ \\\hline
\end{tabular}
}
}
\end{table}

As illustrated in Example~\ref{eg1}, the answers given by each server are generated as sums of the three kinds of columns. Because of the record symmetry enforced by applying (a2) throughout the scheme, for any ${\rm Serv}^{(i)},i\in[N]$ and all $1\leq j\leq M$, each type of $j$-sums appears for the same number of times, that is, for any $\Lambda\subseteq[M]$ with $|\Lambda|=j$, the number of $\Lambda$-type sums provided by ${\rm Serv}^{(i)}$ only depends on $i$ and $j$. We denote this number by $\gamma^{(i)}_j$.  For example, in Example~\ref{eg2} we have $\gamma_1^{(1)}=\gamma_2^{(1)}=2,~\gamma_3^{(1)}=0$ and $\gamma_1^{(i)}=3,~\gamma_2^{(i)}=\gamma_3^{(i)}=1$ for $i=2,3$.

The key idea  in  minimizing the sub-packetization in this work is that we abandon the symmetry across all servers enforced in \cite{Bana&Uluk16:CapacityPIRCoded} and instead adapt partial symmetry among the servers. Specifically, we divide the $N$ servers into two groups, the first $N-K$ servers in one group and the remaining $N-K$ servers in the other. We then only enforce the symmetry across the servers within each group. Consequently, we further define notations $\alpha_j$ and $\beta_j$ such that
{\small $\alpha_j \triangleq\gamma_j^{(i)} {\rm ~for~}1\leq i\leq N-K,~~~{\rm and~~}
  \beta_j \triangleq\gamma_j^{(i)} {\rm ~for~}N-K< i\leq N\;.$}
Therefore, a general description of the query sequences can be displayed in
Table~\ref{tab0}, where ${ q}^{(i)}_{\Lambda,h}$ is the $h$th $\Lambda$- type $|\Lambda|$-sum provided the ${\rm Serv}^{(i)}$.
\begin{table}[ht]
\centering
\caption{}\label{tab0}
\setlength{\abovecaptionskip}{0.05cm}
\setlength{\belowcaptionskip}{-0.05cm}
\resizebox{8cm}{!}{
{\footnotesize
\begin{tabular}{|c|l|l|}
  \hline
   $\Lambda$-type& ${\rm Serv}^{(i)},~1\leq i\leq N-K$ & ${\rm Serv}^{(i)},~N-K< i\leq N$\\\hline
  \makecell{$\forall\Lambda\subseteq[M]$\\}&\makecell{${ q}^{(i)}_{\Lambda,1}$\\$\vdots$\\${ q}^{(i)}_{\Lambda,\alpha_{|\Lambda|}}$}&\makecell{${ q}^{(i)}_{\Lambda,1}$\\$\vdots$\\${ q}^{(i)}_{\Lambda,\beta_{|\Lambda|}}$}
  \\\hline
\end{tabular}
}
}
\end{table}

For simplicity, we denote $\underline{\Lambda}=\Lambda\cup\{\theta\}$ for any subset $\Lambda\subseteq[M]-\{\theta\} $ and arrange all the types $\Lambda\subseteq [M]$ in the following order:
 {\small \begin{equation}\label{order}\begin{pmatrix}
\Lambda_{j,1},&\cdots,&\Lambda_{j,r_j},&\cdots,&\underline{\Lambda_{j,1}}&\cdots&\underline{\Lambda_{j,r_j}}
\end{pmatrix}_{0\leq j<M}
\end{equation}}
where $r_j=\binom{M-1}{j}$ and $\{\Lambda_{j,1},...,\Lambda_{j,r_j}\}=\{\Lambda\subseteq [M]-\{\theta\}\mid |\Lambda|=j\}$. Evidently, $r_0=1$ and $\Lambda_{0,1}=\emptyset$. Moreover, for $0 \leq j\leq M-1$ and $h\in[r_j]$, there are $\gamma^{(i)}_j$  $\Lambda_{j,h}$-type $j$-sums and  $\gamma^{(i)}_{j+1}$ $\underline{\Lambda_{j,h}}$-type $(j+1)$-sums downloaded from ${\rm Serv}^{(i)}$.  To specifically explain how these sums are formed, we first define two functions ${\rm Dist}_1(\Lambda,\underline{\Lambda})$ and ${\rm Dist}_2(\Lambda,\Gamma)$.

$(1)$ The function  ${\rm Dist}_1(\Lambda,\underline{\Lambda})$ for all $\Lambda\subseteq[M]-\{\theta\}$
generates the $\Lambda$-type parts in all $\underline{\Lambda}$-type sums. Its operation follows the following rules:
\begin{itemize}
\item[({\bf b1})] The $\Lambda$-type parts all come from the $\Lambda$-type sums provided by the servers.
\item[({\bf b2})] For each ${\rm Serv}^{(i)}$, $i\in[N]$, its $\Lambda$-type sums and the $\Lambda$-type parts in all its $\underline{\Lambda}$-types sums are distinct.
\end{itemize}

For example, in Example \ref{eg2} where $\theta=1$, let $\Lambda=\{2\}$, then the result of ${\rm Dist}_{1}(\{2\},\{1,2\})$ is displayed in the left table in Figure~\ref{fg44}. One can  observe that $\underline{\bf b}_3,\underline{\bf b}_4$ are exactly the $\{2\}$-type parts of $\rm{Serv}^{(1)}$'s $\{1,2\}$-types sums because it provides the $\{1,2\}$-types sums  $\underline{\bf a}_5+\underline{\bf b}_3, \underline{\bf a}_6+\underline{\bf b}_4$ as shown in Figure~ \ref{fg3}. In addition, one can  verify that ${\rm Dist}_{1}(\{2\},\{1,2\})$ satisfies the rule ({b1}) and ({b2}).

\begin{figure}[ht]
\centering
\begin{tikzpicture}[scale=2]
\node at (0,0){\footnotesize\begin{tabular}{|c|c|c|}
\hline$\rm{Serv}^{(1)}$ & ${\rm Serv}^{(2)}$ & $\rm{Serv}^{(3)}$\\\hline
\makecell{$\underline{\bf b}_3$\\$\underline{\bf b}_4$}&$\underline{\bf b}_2$&$\underline{\bf b}_1$\\
\hline
\end{tabular}};
\node at (0,-0.5){\footnotesize ${\rm Dist}_{1}(\{2\},\{1,2\})$ };

\node at (2.4,0){\footnotesize\begin{tabular}{|c|c|c|}
\hline$\rm{Serv}^{(1)}$ & $\rm{Serv}^{(2)}$ & $\rm{Serv}^{(3)}$\\\hline
\makecell{$\underline{\bf a}_5$\\$\underline{\bf a}_6$}&$\underline{\bf a}_5$&$\underline{\bf a}_6$\\
\hline
\end{tabular}};
\node at (2.4,-0.5){\footnotesize ${\rm Dist}_{2}(\{1\},\{1,2\})$ };

\node at (4.8,0){\footnotesize\begin{tabular}{|c|c|c|}
\hline$\rm{Serv}^{(1)}$ & $\rm{Serv}^{(2)}$ & $\rm{Serv}^{(3)}$\\\hline
\makecell{$\underline{\bf b}_5+\underline{\bf c}_5$\\$\underline{\bf b}_6+\underline{\bf c}_6$}&$\underline{\bf b}_5+\underline{\bf c}_5$&$\underline{\bf b}_6+\underline{\bf c}_6$\\
\hline
\end{tabular}};
\node at (4.8,-0.5){\footnotesize ${\rm Dist}_{2}(\{2,3\},\{2,3\})$};
\end{tikzpicture}
\caption{Explanations of the functions ${\rm Dist}_1(\Lambda,\underline{\Lambda})$ and ${\rm Dist}_2(\Lambda,\Gamma)$ in Example \ref{eg2}.}
\label{fg44}
\end{figure}

$(2)$ The function ${\rm Dist}_2(\Lambda,\Gamma)$ generates the $\Lambda$-type parts in all $\Gamma$-type sums, where $\Lambda=\Gamma\subseteq[M]-\{\theta\}$ or $\Lambda=\{\theta\}\subseteq\Gamma\subseteq[M]$. Its operation follows the following rules:
\begin{itemize}
\item[({\bf b3})] For any $i\in[N]$, ${\rm Serv}^{(i)}$ gets $\gamma^{(i)}_{|\Gamma|}$ $\Lambda$-type parts each of which is allocated to a $\Gamma$-type sums.
\item[({\bf b4})] Each $\Lambda$-type part appears in $K$ different servers.
\end{itemize}

For example, for the case $\theta=1$ in the Example \ref{eg2}, the results of ${\rm Dist}_{2}(\{1\},\{1,2\})$ and
 ${\rm Dist}_{2}(\{2,3\},\{2,3\})$ are displayed in the middle table and the right table in Figure~\ref{fg44}, respectively.
By repeatedly invoking the functions ${\rm Dist}_{1}(\cdot,\cdot)$ and ${\rm Dist}_{2}(\cdot,\cdot)$, the {\bf Algorithm}~\ref{alg: Sym} generates all queries to each server.

\begin{algorithm}
  \caption{}
  \label{alg: Sym}
  \begin{algorithmic}[1]
  \REQUIRE $\theta$
  \ENSURE  $(Q^{(1)}_{\theta},...,Q^{(N)}_{\theta})$
  \STATE  Initialize : $(Q^{(1)}_{\theta},...,Q^{(N)}_{\theta})\leftarrow\emptyset$
   \FOR {$j=0:M-1$}
    \FOR {$h=1:r_j$ }
   \STATE  $(Q^{(1)}_{\Lambda_{j,h}},...,Q^{(N)}_{\Lambda_{j,h}})\leftarrow {\rm Dist}_{2}(\Lambda_{j,h},\Lambda_{j,h})$,
 \STATE
  $(Q^{(1)}_{\underline{\Lambda_{j,h}}},...,Q^{(N)}_{\underline{\Lambda_{j,h}}})\leftarrow {\rm Dist}_{1}(\Lambda_{j,h},\underline{\Lambda_{j,h}})+{\rm Dist}_{2}(\theta,\underline{\Lambda_{j,h}})$
  \FOR {$i=1:N$}
  \STATE
  $Q^{(i)}_{\theta}\leftarrow Q^{(i)}_{\theta}\cup\{Q^{(i)}_{\Lambda_{j,h}},Q^{(i)}_{\underline{\Lambda_{j,h}}}\}$
  \ENDFOR
   \ENDFOR
  \ENDFOR
  \end{algorithmic}
\end{algorithm}

Now let us discuss the way of realizing the functions ${\rm Dist}_{1}(\cdot,\cdot)$ and ${\rm Dist}_{2}(\cdot,\cdot)$.  Note that for the function ${\rm Dist}_{2}(\cdot,\cdot)$, a necessary condition of the requirement $\rm (b3),(b4)$ is
\begin{equation}\label{C1}
K|(N-K)\alpha_j+K\beta_j ~~{\rm for} ~~1\leq j \leq M.
\end{equation}
Moreover, for the function ${\rm Dist}_{1}(\cdot,\cdot)$, a necessary condition of the requirement $\rm (b1),(b2)$ is
\begin{equation}\label{C2}{\small
\left\{\begin{array}{l}
\alpha_{j+1}+\alpha_j=\frac{(N-K)\alpha_j+K\beta_j}{K},\\
\beta_{j+1}+\beta_j=\frac{(N-K)\alpha_j+K\beta_j}{K},\\
 \alpha_j,\beta_j\in\mathbb{N}, j\in[M].
\end{array}\right.}
\end{equation}

Assume that we have determined the values of $\alpha_j$ and $\beta_j$ such that (\ref{C1}) and (\ref{C2}) are satisfied. Then, the function ${\rm Dist}_{1}(\Lambda,\underline{\Lambda})$ can be realized as follows. Suppose $|\Lambda|=j$.  Denote $Q_{\Lambda}$ as the set of all $\Lambda$-type $j$-sums contained in all servers. $Q_{\Lambda}$ is actually the output of the function ${\rm Dist}_{2}(\Lambda,\Lambda)$. Therefore, $|Q_{\Lambda}|=\frac{\sum^N_{i=1}\gamma^{(i)}_j}{K}=\frac{(N-K)\alpha_j+K\beta_j}{K}$.
Let $Q^{(i)}_{\Lambda},i\in[N]$ be the set of $\Lambda$-type $j$-sums contained in the server ${\rm Serv}^{(i)}$, i.e., that is the $i$-th tuple of ${\rm Dist}_{2}(\Lambda,\Lambda)$. Then, $|Q^{(i)}_{\Lambda}|=\alpha_j$ for  $1\leq i\leq N-K$ and $|Q^{(i)}_{\Lambda}|=\beta_j$ for $N-K+1\leq i\leq N$. By the identity (\ref{C2}), the $i$-th tuple of ${\rm Dist}_{1}(\Lambda,\underline{\Lambda})$ is $Q_{\Lambda}-Q^{(i)}_{\Lambda}$ for $1\leq i\leq N$.

 The key point in realizing the function ${\rm Dist}_{2}(\Lambda,\Gamma)$ is to ensure  symmetry across the servers within each group and to simultaneously satisfy the rule ({b4).  Note that for the last $K$ servers, the requirement (b4) and symmetry of servers can be easily satisfied, and therefore we only need to consider the first $N-K$ tuples of ${\rm Dist}_{2}(\Lambda,\Gamma)$. However, when $N-K<K$, the first $N-K$ tuples of ${\rm Dist}_{2}(\Lambda,\Gamma)$ cannot locally satisfy (b4) in the group, and so it needs help from the last $K$ tuples.  Therefore, we realize the function ${\rm Dist}_{2}(\Lambda,\Gamma)$ for the case $N-K\geq K$ and $N-K<K$ separately.  We first define an index function ${\rm IniCol}(\Lambda,\Gamma)=(\ell_j(\Lambda,\Gamma))_{j\in\Lambda}$, which returns the initial index of the columns used in the function ${\rm Dist}_{2}(\Lambda,\Gamma)$ when it is invoked in the {\bf Algorithm}~\ref{alg: Sym}. Hence, ${\rm IniCol}(\emptyset,\emptyset)=0$ and ${\rm IniCol}(\{j\},\{j\})=1$ for $ j\in[M]$. Moreover for $1<\nu<M, 1\leq h\leq r_\nu$,
 \begin{equation} \begin{split}\label{I1}
  \ell_j(\Lambda_{\nu,h},\Lambda_{\nu,h})&=\sum^{\nu-1}_{s=1}\sum^{r_s}_{t=1}\chi_{\Lambda_{s,t}}(j)\frac{(N-K)\alpha_s
  +K\beta_s}{K}  +\sum^{h-1}_{t=1}\chi_{\Lambda_{s,t}}(j)\frac{(N-K)\alpha_\nu+K\beta_\nu}{K}+\chi_{\Lambda_{\nu,h}}(j),\\
  \ell_{\theta}(\{\theta\},\underline{\Lambda_{\nu,h}})
   &=\sum^{\nu-1}_{s=0}r_{s}\frac{(N-K)\alpha_{s+1} +K\beta_{s+1}}{K}
   +(h-1)\frac{(N-K)\alpha_{\nu+1}+K\beta_{\nu+1}}{K}+1,
  \end{split}
 \end{equation}
 where $j\in[M]-\{\theta\}$ and $\chi_{\Lambda}(\cdot)$ is the characteristic function of the set $\Lambda$, i.e., $\chi_{\Lambda}(a)=1$ if $a\in\Lambda$ and $\chi_{\Lambda}(a)=0$ if $a\notin\Lambda$.
 For example, for the case $\theta=1$ in the Example~\ref{eg2}, ${\rm IniCol}(\{2,3\},\{2,3\})=(5,5)$ and ${\rm IniCol}(\{1\},\{1,2\})=5$. Denote ${\bf q}_{\Lambda,h}=\sum_{j\in\Lambda}u_{\ell_j+h-1}$, where $(\ell_j)_{j\in\Lambda}={\rm IniCol}(\Lambda,\Gamma)$.

  (1) ${\rm Dist}_{2}(\Lambda,\Gamma)$ for the case $N\geq 2K$

  Let $t=\frac{N-K}{K}\alpha_{|\Gamma|}$.
  The first step of ${\rm Dist}_{2}(\Lambda,\Gamma)$ is generating the first $N-K$ tuples, that is, arrange the  $t$ $\Lambda$-type sums ${\bf q}_{\Lambda,1},{\bf q}_{\Lambda,2},...,{\bf q}_{\Lambda,t}$ to the first $N-K$ servers according to  the rules (b3),(b4), as displayed in Table~ \ref{tab}.
\begin{table}[ht]
\centering
\caption{}\label{tab}
{\footnotesize
\begin{tabular}{l}
$\overbrace{{\rm Serv}^{(1)}{\rm Serv}^{(2)} ~\cdots ~\cdots ~ {\rm Serv}^{(N-K)}}^{N-K}$ ~${\rm Serv}^{(1)}$~~$\cdots$ \\
  $\underbrace{{\bf q}_{\Lambda,1}~~~~{\bf q}_{\Lambda,1}~\cdots~{\bf q}_{\Lambda,1}}_{K}$~$\underbrace{{\bf q}_{\Lambda,2}~~~~{\bf q}_{\Lambda,2}~\cdots~~{\bf q}_{\Lambda,2}}_{K}$~$\cdots$
\end{tabular}
}
\end{table}
 Therefore, for the server ${\rm Serv}^{(i)},1\leq i\leq N-K$, the $h$-th component of the $i$-th tuple of ${\rm Dist}_{2}(\Lambda,\Gamma)$ is ${\bf q}_{\Lambda, \lceil\frac{(h-1)(N-K)+i}{K}\rceil}$ for $1\leq h\leq \alpha_{|\Gamma|}$.

 The second step of ${\rm Dist}_{2}(\Lambda,\Gamma)$ is  generating the last $K$ tuples such that (b3) and (b4) are satisfied. The generating process is the same as that of the first step. Note that there are exactly $K$ servers, then each tuple is ${\bf q}_{t+1},{\bf q}_{t+2},...,{\bf q}_{t+\beta_{|\Gamma|}}$.

 (2) ${\rm Dist}_{2}(\Lambda,\Gamma)$ for the case $N<2K$

 ${\rm Dist}_{2}(\Lambda,\Gamma)$ generates the first $N-K$ tuples, each of which has the form ${\bf q}_{\Lambda,1},{\bf q}_{\Lambda,2},...,{\bf q}_{\Lambda,\alpha_{|\Gamma|}}$. Note that each $\Lambda$-type sum ${\bf q}_{\Lambda,h}$ does not satisfy the requirement $(b4)$,  because it only appears in $N-K$ servers at present. We must then use $K-(N-K)=2K-N$ out of the last $K$ servers to provide the additional copies of these $\Lambda$-type sums. Therefore, we symmetrically allocate the $\Lambda$-type sums ${\bf q}_{\Lambda,1},{\bf q}_{\Lambda,2},...,{\bf q}_{\Lambda,\alpha_{|\Gamma|}}$ to the last $K$ servers such that each $\Lambda$-type sum appears in $2K-N$ servers. Note that $K\geq 2K-N$ and $K|(2K-N)\alpha_{|\Gamma|}$  from (\ref{C1}). Similarly, as the map displayed in the Table~ \ref{tab}, we can realize such an arrangement and each of the last $K$ servers provide $\frac{2K-N}{K}\alpha_{|\Gamma|}$ $\Lambda$-type sums. Finally, to satisfy the rule (b3), we need to send another $\gamma_{|\Gamma|}\triangleq\beta_{|\Gamma|}-\frac{2K-N}{K}\alpha_{|\Gamma|}$
$\Lambda$-type sums to each of the last $K$ tuples. Note that $\gamma_{|\Gamma|}\geq 0$ from (\ref{C2}).

\subsection{Parameters in the scheme}\label{sec5c}
In this section, we determine the values of $\alpha_j,\beta_j, 1\leq j\leq M$ such that (\ref{C1}) and (\ref{C2}) are satisfied.
Suppose $d={\rm gcd}(N,K),n=\frac{N}{d},k=\frac{K}{d}$. From (\ref{C1}) and (\ref{C2}), we obtain the following equations,
\begin{equation}\label{G1}
\left\{\begin{array}{l}
\beta_{j+1}=\frac{n-k}{k}\alpha_{j},\\
\alpha_{j+1}=\beta_j+\frac{n-2k}{k}\alpha_{j},
\\k|(n-k)\alpha_M,\\
\alpha_j,\beta_j\in \mathbb{N},~1\leq j\leq M.
\end{array}\right.
\end{equation}
From the recursive relations in (\ref{G1}), we obtain two geometric series, i.e.,
\begin{equation}\label{eqser2}{\small
\left\{\begin{array}{l}
(n-k)\alpha_{j+1}+k\beta_{j+1}=\frac{n-k}{k}((n-k)\alpha_j+k\beta_j),\\
\alpha_{j+1}-\beta_{j+1}=-(\alpha_j-\beta_j),1\leq j< M.
\end{array}\right.}
\end{equation}

By assigning proper values to $(\alpha_1,\beta_1)$ or $(\alpha_M,\beta_M)$, we can obtain the integer solutions of (\ref{eqser2}) and in turn get the solutions of (\ref{G1}).
When $N\geq 2K$, set  $\alpha_1=0,\beta_1=k^{M-1}$, we obtain the integer solutions \begin{equation}\label{eqsolution1}
\left\{\begin{array}{l}\alpha_j= \frac{(n-k)^{j-1}-(-k)^{j-1}}{n}k^{M-j+1}, \\
    \beta_j = \frac{(n-k)^{j-2}-(-k)^{j-2}}{n}(n-k)k^{M-j+1}.
\end{array}\right.
\end{equation}
When $N<2K$, set $\alpha_M=0,\beta_M=(n-k)^{M-1}$, we obtain the integer solutions
{\small \begin{equation}\label{eqsolution2}
\left\{\begin{array}{l}\alpha_j=\frac{k^{M-j}-(k-n)^{M-j}}{n}k(n-k)^{j-1}, \\
    \beta_j = \frac{k^{M-j+1}-(k-n)^{M-j+1}}{n}(n-k)^{j-1}.
\end{array}\right.
\end{equation} }
It is easy to verify that (\ref{eqsolution1}) and (\ref{eqsolution2}) provide the solution to (\ref{G1}) for the case $N\geq 2K$ and for the case $K<N<2K$, respectively. In both cases, we can observe that  \begin{equation}\label{eqeq}
\frac{(N-K)\alpha_{j}+K\beta_{j}}{K}=(n-k)^{j-1}k^{M-j}.\;
\end{equation}

 We then calculate the parameter
$\tilde{L}$. For the desired record $W_\theta$,  one can observe from (\ref{I1}) that
\begin{equation}\begin{split}\label{L1}
 \tilde{L}&\geq\ell_{\theta}(\{\theta\},[M])-1+\frac{(N-K)\alpha_{M}+K\beta_M}{K}=n^{M-1}.
 \end{split}
 \end{equation}
However, for each undesired record $U_j,j\in[M]-\{\theta\}$, it has
\begin{equation}\begin{split}\label{L2}
 \tilde{L}&\geq\ell_j([M]-\{\theta\},[M]-\{\theta\})-1+\frac{(N-K)\alpha_{M-1}+K\beta_{M-1}}{K}=kn^{M-2}.
 \end{split}
\end{equation}
Therefore, it is sufficient to set $\tilde{L}=\max\{n^{M-1},kn^{M-2}\}=n^{M-1}$, and thus the sub-packetization $L=K\tilde{L}=Kn^{M-1}$ in our scheme. Note that in (\ref{L1}) and (\ref{L2}) we respectively calculate the number of columns from the desired record and from each undesired record invoked by all $N$ servers, and it turns out that the former is larger than the latter. However, for each individual server, the number of columns from each record are invoked because of the record symmetry within each server. Finally, we calculate the access number of our scheme, which is
\begin{equation}\begin{split}\label{L3}
  \omega=\sum^N_{i=1}\sum^M_{j=1}\binom{M}{j}j\gamma^{(i)}_j
 =MK\sum^{M}_{j=1}\binom{M-1}{j-1}(\frac{N-K}{K}\alpha_{j}+\beta_{j})
 =MKn^{M-1}.
 \end{split}
 \end{equation}

\subsection{Properties of the scheme}\label{sec5d}
 Because the sub-packetization and the access number has been calculated in the last section for simultaneously attaining the lower bounds, we must still
verify that the scheme described in Section \ref{sec5b} satisfies the correctness and the privacy condition and achieves the capacity.

The correctness condition can be easily verified based on the rules (b1), (b4), and the $[N,K]$ MDS code which is used in distributed storage system. The privacy condition is derived from the rules (b2) and (b3), which implies that for each individual server, the same number of columns from each record are invoked in a symmetric form and each column is invoked at most one time.

Finally, from Table~\ref{tab0}, we can compute the download size, i.e.,
  \begin{align*}
  D=\sum^M_{j=1}\binom{M}{j}((N-K)\alpha_j+K\beta_j)
  \stackrel{(a)}{=}K\sum^{M-1}_{j=0}\binom{M}{j}(n-k)^{j-1}k^{M-j}
  =K\frac{n^{M}-k^M}{n-k},
  \end{align*}
where ${\small (a)}$ comes from the identity (\ref{eqeq}). One can immediately check that the rate of our scheme achieves the capacity for coded PIR.
\begin{corollary}
For $M\geq 2$ and $1\leq K<N$, the optimal sub-packetization and the optimal access number for linear capacity-achieving coded PIR schemes from MDS coded non-colluding servers are $Kn^{M-1}$ and $MKn^{M-1}$ respectively, where $d={\rm gcd}(N,K)$ and $n=N/d$.
\end{corollary}

\section{Conclusions}\label{sec6}

In this paper, we investigated the problem of minimizing the sub-packetization and the access number for all linear capacity-achieving PIR schemes from $[N,K]$ MDS coded non-colluding servers. The optimal sub-packetization and the optimal access number are explicitly determined for all nontrivial cases. The process of proving the lower bound on the sub-packetization is an extension of the proof of \cite{Zhang&Xu17:OptimalSubpacketization} from replicated PIR to coded PIR, and the main design idea for reducing the sub-packetization by using partial symmetry  among servers is the same as that used in \cite{Zhang&Xu17:OptimalSubpacketization}. However, some extra proof skills and design rules are specially developed  for coded PIR, such as the function Vec introduced in Section 3.2 and the design rule (b4) in Section 5.2. In addition, our approach for characterizing the minimum sub-packetization and the minimum access number can be extended to other PIR models.

\end{document}